\documentclass[conference]{IEEEtran}
\IEEEoverridecommandlockouts
\usepackage{cite}
\usepackage{amsmath,amssymb,amsfonts}
\usepackage{algorithmic}
\usepackage{graphicx}
\usepackage{textcomp}
\usepackage{xcolor}
\usepackage{romannum}
\usepackage{amsthm}
\usepackage{algorithmic}
\usepackage[ruled,vlined,linesnumbered]{algorithm2e}
\SetAlFnt{\footnotesize}
\def\BibTeX{{\rm B\kern-.05em{\sc i\kern-.025em b}\kern-.08em
T\kern-.1667em\lower.7ex\hbox{E}\kern-.125emX}}

\newtheorem{theorem}{Theorem}
\newtheorem{lemma}{Lemma}
\newtheorem{definition}{Definition}
\begin{document}

\title{Middle-mile Network Optimization \\ in Rural Wireless Meshes}

\author{\IEEEauthorblockN{Yung-Fu Chen}
\IEEEauthorblockA{\textit{Dept. of Computer Science and Engineering} \\
\textit{The Ohio State University}\\
Columbus, OH, USA \\
chen.6655@osu.edu}
\and
\IEEEauthorblockN{Anish Arora}
\IEEEauthorblockA{\textit{Dept. of Computer Science and Engineering} \\
\textit{The Ohio State University}\\
Columbus, OH, USA \\
anish@cse.ohio-state.edu}
}

\maketitle

\begin{abstract}
The status quo of limited broadband connectivity in rural areas motivates the need for fielding alternatives such as long-distance wireless mesh networks. A key aspect of fielding wireless meshes cost-effectively is planning how to connect the last-mile networks to the core network service providers (i.e., the network between the edge access terminals and the landline / optical fiber terminals) with minimal infrastructure cost and throughput constraints. This so-called middle-mile network optimization, which includes topology construction, tower height assignment, antenna and orientation selection, as well as transmit power assignment, is known to be a computationally hard problem.

In this paper, we provide the first polynomial time approximation solution for a generalized version of the middle-mile network optimization problem, wherein point-to-point (i.e., WiFi p2p) links are deployed to bridge last-mile networks. Our solution has a cost performance ratio of $O(\ln{|A|}+\frac{|B|}{|A|}+\frac{|A|+|B|}{\gamma})$, where A and B respectively denote the number of terminals and non-terminals and $\gamma$ is the ratio of $\frac{link\ capacity}{terminal\ demand}$. Furthermore, our solution extends to hybrid networks, i.e., point-to-multipoint (i.e., WiFi p2mp) or omnidirectional (i.e., TV White Space) can serve as hyperlinks in addition to point-to-point links, to further reduce the cost of wireless links. We provide a complementary heuristic for our middle-mile network optimization solution that adds hyperlinks if and only if they reduce the cost.
\end{abstract}


\section{Introduction}
Even as the demand for broadband escalates worldwide to support real-time Internet services such as streaming, distance learning and telehealth, rural areas tend to be underserved if not unserved.  In the US alone, of the 24 million Americans who lack access to fixed broadband services of at least 25 Mbps/3 Mbps as of 2016, the vast majority live in rural areas: 31.4\% of the population in rural areas versus 2.1\% in urban areas.  And while fiber optic cable deployment continues, its high cost implies that wireless backbones remain a necessary part of the addressing the broadband digital divide for the foreseeable future.  


Among the wireless approaches, wireless mesh networks offer a resilient alternative for spanning tens to a hundred kilometers in rural areas. They offer the potential of being low-cost, in part because they use unlicensed spectrum, unlike cellular network. IEEE 802.11-based long-distance wireless mesh networks have therefore been widely deployed to fill rural broadband gaps worldwide, \cite{chebrolu2006long, airjaldi, matthee2007bringing}. Directional antennas of small beamwidth with WiFi radio are mounted on the top of towers of height sufficient to line-of-sight propagation over the links. TV White Space (TVWS) meshes are also starting to proliferate \cite{msairband, mspartnerships} as a low-cost alternative since their spectrum offers non-line-of-sight propagation over foliage and obstructions, as compared to the 2.4GHz and 5GHz spectrum used in many IEEE 802.11 meshes. This allows TVWS to use lower height towers than WiFi, which is particularly cost-efficient for hilly rural areas. Also, the IEEE 802.22 TVWS range is competitive with that of 802.11 long-distance wireless links: a base station can provide the transmission radius up to 100 kilometers  \cite{cordeiro2005ieee,stevenson2009ieee}. 

A major challenge for wireless meshes is sustainability. The growth of demand for their services versus their profitability is a chicken-and-egg problem, which is made harder by a lack of networking planning tools that support incremental deployment.  Part of the challenge is to contain the deployment cost, which can be dominated by tower costs.  Towers are typically located at or near population hubs that need network connectivity, but suitable towers/structures are not easily available or afforded to providers.  Tower cost is height dependent (say \$100-\$5000 for towers from 10-45m). Communication equipment cost is also a factor: while WiFi equipment is comparatively cheap (say \$50 for WiFi p2p antennas), TVWS equipment
is yet to enjoy economy of scale.

Anecdotally, our experience in dealing with rural broadband wireless providers corroborate the challenge in cost, throughput assurance, and flexibility for network growth. In a project \cite{bdspoke} on infrastructure supporting rural counties in southern Ohio for data-driven services and analytics related to the opioid epidemic, we interacted with providers to deploy WiFi, TVWS, or hybrid meshes that would be leveraged by project stakeholders among other users.  Providers point to high costs incurred for tower access, socio-political difficulty of leveraging community-shared tower structures for cost reduction, and networks planning that tends to err on the side of wasted capacity or does not efficiently support incremental growth of the networks. 

\textbf{Wireless Mesh Planning and MNO.}  Network planning in wireless meshes can be divided into the {\em last-mile network optimization (LNO)}  problem and the {\em middle-mile network optimization (MNO)}  problem. Both are NP-hard \cite{aoun2006gateway,sen2007long}. The former problem involves placing a minimum number of edge access terminals (i.e., gateways) to collect traffic flows from home APs while satisfying their QoS requirements. A polynomial time LNO approximation algorithm \cite{aoun2006gateway} exists that recursively calculates minimum weighted Dominating Sets (DS) while guaranteeing QoS requirements in each iteration. 

MNO seeks to connect the edge-access terminals to core network service providers with minimal infrastructure cost, while satisfying the throughput constraints of each edge-access terminal. The MNO problem was first formulated \cite{sen2007long} as: 
\begin{quote}
\textit{Given a set of terminals to be connected to a given landline (and a set of non-terminals with existing towers to be used as relay vertices), determine the network topology, the tower heights of terminals, antenna types and orientations, transmit powers, and the route for each terminal to the landline, such that throughput constraints of terminals are satisfied and the total cost of towers and antennas is minimized.}
\end{quote}
MNO problem is not yet well solved:  Its solution in \cite{sen2007long} has exponential time complexity, due to the formulation of height assignment and transmit power assignment as Linear Programming (LP) problems. Although a greedy approach has been proposed \cite{panigrahi2008minimum} for the NP-hard subproblem of topology construction along with tower height assignment, which has an approximation ratio of $O(\log{n})$, we are not aware of approximation algorithms for MNO problem per se. 


\textbf{Contributions of the paper.} In this paper, we design an approximation algorithm to solve MNO problem.  To that end, we introduce the problem of Topology Construction of Minimum Steiner Tree (SteinerTC) problem.\footnote{The SteinerTC problem---which is derived from the Topology Construction (TC) problem \cite{panigrahi2008minimum}---incorporates topology construction and tower height assignment in MNO problem.} Then, we solve MNO problem in two parts:  solving first the SteinerTC problem and  solving thereafter the Capacitated Network Design (CND) problem \cite{hassin2004approximation}. Next, we analytically evaluate the worst-case cost bound of the solution. Lastly, we extend the algorithm for solving MNO problem in hybrid networks, that allow incorporation of point-to-multipoint (i.e., WiFi p2mp) and omnidirectional links (i.e., TVWS) as hyperlinks. We give a complementary heuristic that replaces p2p links by p2mp or omnidirectional hyperlinks if and only if they reduce the total cost.

The contribution of the paper is four-fold:  First, ours is the first polynomial-time approximation algorithm for MNO problem. Second, we show the cost performance ratio of our MNO solution is  $O(\ln{|A|}+\frac{|B|}{|A|}+\frac{|A|+|B|}{\gamma})$, where $|A|$ is the number of terminals, $|B|$ is the number of non-terminal vertices and $\gamma$ is the ratio of $\frac{link\ capacity}{terminal\ demand}$.  Third, for the sub-problem of SteinerTC, which we show is NP-hard  (Theorem~\ref{npSteinerTC}) by reducing from TC problem, we give a greedy solution with an approximation ratio of $O(\ln{|A|})$  for the tower cost (Theorem~\ref{apporxSteinerTC}).  Finally, our MNO solution for hybrid networks achieves Capacity Flexibility, i.e., the added hyperlinks yield residual capacity that can be leveraged to add additional edge access terminals in the future without incurring additional tower cost.

\textbf{Organization.} In Section \Romannum{2}, we discuss the related work on TC and CND problems. In Section \Romannum{3}, we formalize the MNO problem. We then present our approximation solution for a homogeneous case (p2p links only) of MNO problem and, in Section \Romannum{4}, analyze the cost bound. We describe, in Section \Romannum{5}, our complementary heuristic that further reduces the infrastructure cost by introducing hyperlinks, and make concluding remarks, in Section \Romannum{6}.

\section{Related Work}
\subsection{Topology Construction (TC)}
The TC problem is to determine a network topology that spans all the given terminals as well as a minimal height of antenna tower of each terminal such that line-of-sight propagation using a WiFi p2p link exists for each edge. TC was formalized as follows \cite{panigrahi2008minimum}: 
\begin{quote}
\textit{Given an undirected graph G = (V, E), obstruction locations and their heights on each edge in G, and a cost function of tower height, determine a set of tower heights at vertices in V such that the subgraph induced by the edges covered by the set of tower heights is a subgraph (i.e., a spanning tree) of G spanning all vertices such that the total cost of the towers is minimized.} 
\end{quote}
The TC problem is known to be NP-hard and solvable by iteratively choosing the height increments of a vertex and a set of its neighbors with the best cost-to-benefit ratio until a spanning tree is determined.

Note that this TC formalization only considers topologies that span all vertices, i.e., each vertex is a terminal. One way of reducing the cost is to leverage pre-existing towers in rural areas to host non-terminal vertices (i.e., relays), thus eschewing additional tower cost while reducing the tower cost of some terminals. With this approach, the problem of determining the minimum cost subgraph of G that connects all terminals and may include some non-terminals becomes that of computing the Minimum Steiner Tree. 
This leads us to pose the Topology Construction problem for Minimum Steiner Tree (SteinerTC) problem of determining a set of tower height of terminals forming a Steiner tree with minimized cost. In Section \Romannum{3} we show that SteinerTC is at least NP-hard and design an algorithmic solution for it. 

\subsection{Capacitated Network Design (CND)}
Also related to our work is the CND problem of installing adequate capacity for meeting the bandwidth needs between a landline and edge access terminals. CND was formalized as follows \cite{salman2001approximating, hassin2004approximation}: 
\begin{quote}\textit{Given an undirected graph G = (V, E), a set of terminals with their traffic demands, a landline, and uniform link capacity to be installed on edges, we are going to determine a minimum cost of capacity installation in G so that the installed capacity can route all traffics from terminals to the landline.} 
\end{quote}
The authors in \cite{hassin2004approximation} propose an approximation algorithm by introducing a relation between the CND problem and minimum Steiner tree problem. First, they take advantage of known solutions to Steiner tree problem \cite{takahashi1990approximate} to obtain a Steiner tree with minimum cost. Second, their solution divides terminals into groups whose internal traffic flows are equal to or less than the cable capacity, and assign a hub in each group to forward the aggregated traffic from the corresponding terminals to the landline through the additionally installed cable capacity\footnote{An edge in a graph can be installed more than one copy of cable.}, if necessary, on the shortest path between the hub and landline. Since the CND problem satisfies throughput constraints by finding a minimum cost of wired cable installation in a given graph instead of considering a wireless case where line-of-sight propagation is necessary to support the transmission of WiFi spectrum, the CND problem is only a particular instance of MNO problem. A solution to CND problem does not comprehensively solve MNO problem.


\section{Problem Statement}
In this section we describe the network model and state the middle-mile network optimization problem, whose goal, in addition to satisfying throughput constraints of terminals, is to minimize the cost infrastructure deployment consisting of towers and antennas, in a general way. The main notations are summarized in Table~\ref{notationMNO}. But first, we introduce the problem of SteinerTC.

\subsection{Topology Construction of Minimum Steiner Tree (SteinerTC) Problem and Hardness}
Unlike the TC problem, SteinerTC introduces non-terminals provided with towers as relays so as to minimize tower heights of terminals. We formulate the SteinerTC problem as follows: 
\begin{quote}
\textit{Given an undirected graph $G' = (V', E')$, a set of terminals $A \subset V'$, a set of non-terminals $B \subset V'$, $A \cup B = V'$, obstruction locations and their heights on each edge in $G'$, and a cost function of tower height, we are going to determine a set of tower heights at vertices in $A$ such that the subgraph induced by the edges covered by the set of tower heights in $V'$ is a subgraph (i.e., a Steiner tree) of $G'$ spanning all vertices in $A$ such that the total cost of the towers at terminals is minimized.} 
\end{quote}
Compared to the output graph of TC problem, which is a minimum cost spanning tree where the induced vertices are all terminals, the output graph of SteinerTC is a minimum cost Steiner tree where the spanned vertices consist of all terminals and some of the non-terminals. 

We give a brief proof of the following theorem by reducing TC to SteinerTC.

\begin{theorem} 
SteinerTC is at least as hard as TC, which is NP-hard.
\label{npSteinerTC}
\end{theorem}

\begin{proof}
For any two neighbor vertices $u$ and $v$ in the graph $G$ in TC problem, assume we are allowed to create at least one intermediate vertex with fixed tower height between $u$ and $v$ but zero cost if there exists an edge $(u, v)$ of line-of-sight propagation associated with height increments for $u$ and $v$, denoted as $h^+\!(u)$ and $h^+\!(v)$ respectively. Accordingly, based on $h^+\!(u)$ and $h^+\!(v)$ which make edge $(u, v)$ with cost increment $cT\!ower^+\!(h^+\!(u))$ and $cT\!ower^+\!(h^+\!(v))$ in TC problem, we could create one or more intermediate vertices $i_1,\dots , i_r$ with fixed heights $h(i_1),\dots ,h(i_r)$ which make the same cost increment, $cT\!ower^+\!(h^+\!(u)')+ cT\!ower^+\!(h^+\!(v)')= cT\!ower^+\!(h^+\!(u)) + cT\!ower^+\!(h^+\!(v))$, to make edges $(u, i_1)$ and $(i_r, v)$ of line-of-sight propagation with heights $h^+\!(u)'$ and $h^+\!(v)'$.

Note that $h(i_1),\dots ,h(i_r)$ are sufficient to make line-of-sight propagation over edges $(i_1, i_2),\dots ,(i_{r-1}, i_r)$. Thus we can say that for any new paths between $(u, v)$, there exist a path with the minimal cost equal to the cost of $cT\!ower^+\!(h(u)) + cT\!ower^+\!(h(v))$. Some intermediate vertices with fixed tower heights could be reused as relays in a new path between any two neighbor vertices in G, thus the degree of intermediate vertex is $\geq$ 2.

Hence, for an instance of SteinerTC problem with the graph $G'$ = ($V'$, $E'$), there exists an instance of TC problem with the graph $G$ = ($V$, $E$), which can be transformed to $G'$, where all the intermediate vertices are non-terminals in $G'$.
\end{proof}

\begin{table}[htbp]
\caption{Notations for Middle-mile Network Optimization}
\begin{center}
\begin{tabular}{c|c}
\hline
\textbf{Symbol}&\textbf{Meaning} \\
\hline
$G$ &input graph of MNO problem\\
&$G$ = ($V$, $E$), $V = A \cup B$ \\
\hline
$A$&set of terminals \\
\hline
$B$ &set of non-terminals \\
\hline
$LN \in A$&landline (core network service provider) \\
\hline
$U \in \mathbb{N}$& capacity of wireless p2p link \\
\hline
$cT\!ower$& tower cost function \\
&$cT\!ower: [H\!T\!M\!I\!N, H\!T\!M\!A\!X] \rightarrow \mathbb{R}$\\
\hline
$H\!T\!M\!I\!N\! \in \!\mathbb{R}$&minimum tower height \\
\hline
$H\!T\!M\!A\!X\! \in \!\mathbb{R}$&maximum tower height \\
\hline
$cLink$& link cost function \\
&$cLink: \mathbb{N} \rightarrow \mathbb{R}$\\
\hline
$dem$&throughput demand function \\
&$dem: V \rightarrow [0, U]$, $\forall v \in B, dem(v)=0$ \\
\hline
$R \in \mathbb{R}$&maximal transmit distance of WiFi p2p links \\
\hline
$ob$&normalized obstruction height function \\
&$ob: E \rightarrow \mathbb{R}$ \\
\hline
$G_{SteinerT\!C}$&output graph of MNO problem \\
&$G_{SteinerT\!C} = (V_{SteinerT\!C}, E_{SteinerT\!C})$ \\
&$A \subseteq V_{SteinerT\!C} \subseteq V$ \\
\hline
$h$&tower height function \\
&$h: V \rightarrow [H\!T\!M\!I\!N, H\!T\!M\!A\!X]$ \\
&$\forall v \in B$, $h(v)$ is given and fixed \\
\hline
$rt$&route function from terminals to $LN$ \\
&$rt: V \rightarrow P$ \\
&$P$: set of paths composed of edges in $E_{SteinerT\!C}$ \\ 
\hline
$f$ &traffic flow function\\
&$f: E_{SteinerT\!C} \rightarrow \mathbb{R}$\\
\hline
$linkset_{M\!P}$& hyperlink function at vertices with $M\!P$ antennas \\
&$linkset_{M\!P}: V_{SteinerT\!C} \rightarrow C\!O\!N\!F\!S$ \\ 
&$C\!O\!N\!F\!S$: set of subset of $M\!P$ antenna configurations \\
&(see Algorithm~\ref{MP-ANT-REPLACE}) \\
\hline
$linkset_{O\!m\!n\!i}$&hyperlink function at vertices with $O\!m\!n\!i$ antennas \\
&$linkset_{O\!m\!n\!i}: V_{SteinerT\!C} \rightarrow VS$\\
&$V\!S$: set of subsets of $V_{SteinerT\!C}$ \\
\hline
\end{tabular}
\label{notationMNO}
\end{center}
\end{table}

\subsection{Reformulation of MNO}
Recall that the parameters to be determined in the MNO problem are network topology, tower heights, antenna type and orientations, transmit powers, and traffic routes. The SteinerTC problem---finding a minimum cost Steiner tree by identifying a set of tower heights---determines the first two of these parameters. Given this Steiner tree as an input graph, the CND problem is able to find a minimum capacity installation routing traffics to landline. This satisfies throughput constraints in a consideration of line-of-sight propagation over WiFi spectrum. Note that, in a homogeneous case (p2p links only) of the MNO problem, the parameters of antenna type and orientations and transmit powers can be simply determined on the basis of the edges in the Steiner tree. (In hybrid networks, the schemes of assigning antenna orientations of WiFi p2mp and transmit power are presented in Section \Romannum{5}.) Hence, solving SteinerTC and CND problems in sequence suffices to determine the parameters in the MNO problem. 

The input and output of MNO problem are reformulated as follows.
\begin{quote}
\textit{\textbf{Input}}. An undirected graph $G$ = ($V$, $E$), a set $A \subset V$ of terminals and a set $B \subset V$ of non-terminals, $A \cup B = V$, a landline $LN\! \in \! A$, a cost function of tower height, a cost function of link installation, and throughput demand $dem(v) \leq U$ for each terminal.
\end{quote}
Note that for any two vertices $u$ and $v$, $(u, v)\! \in \! E$, $dist(u, v) \leq R$. For each edge $(u, v)\! \in \! E$, there exists a normalized obstruction height $ob(u, v)$ at the middle of edge $(u, v)$. Line-of-sight propagation is available over a wireless p2p link connecting $u$ and $v$ if the sum of the tower heights at $u$ and $v$ is not less than twice of $ob(u, v)$.
\begin{quote}
\textit{\textbf {Output}}.~A connected graph $G_{SteinerT\!C}\!=\!(V_{SteinerT\!C}$, $E_{SteinerT\!C})$, 
along with a height function $h$ for each terminal in $A$, the route function $rt$ for each terminal in $V$ indicating the routing path from terminals to the landline, the traffic flow function $f$ for each edge in $E_{SteinerT\!C}$, and the hyperlink functions $linkset_{M\!P}$ and $linkset_{O\!m\!n\!i}$ if a hybrid network is constructed, such that throughput demands of terminals are satisfied and the cost of tower and wireless link installation $\sum_{v\in A}{cT\!ower(h(v))}+\sum_{e\in E_{SteinerT\!C}}{cLink(\lceil \frac{f(e)}{U} \rceil)}$ is minimized. 
\end{quote}
Note that $linkset_{M\!P}$ and $linkset_{O\!m\!n\!i}$ denote the set of hyperlinks. For vertices $v$ where p2mp antenna(s) are installed, $linkset_{M\!P}$ specifies the p2mp antenna configurations at $v$ that include the covered vertices and effective direction/range. For vertices $v$ where an omnidirectional antenna is installed, $linkset_{O\!m\!n\!i}$ returns vertices in $v$'s neighborhood that are connected to the omnidirectional antenna at $v$.

\section{Polynomial Time Approximation Algorithm}

\subsection{Solution to SteinerTC Problem}

A greedy algorithm, TC-ALGO, is proposed in \cite{panigrahi2008minimum} to solve TC with a logarithmic cost bound. In each iteration of TC-ALGO, all terminals are considered with different height increments associated with cost increments, using a doubling search, to find the one with the lowest cost-to-benefit ratio. For each combination of a terminal and a certain height increment considered, TC-ALGO considers its neighbor terminals with the minimum height increments to achieve line-of-sight links by calling the subroutine of STAR-TC-ALGO. Its evaluation metric is the cost-to-benefit ratio defined as the minimum increment in cost for the maximum reduction in the number of isolated components. A component of a graph denotes a maximal connected subgraph.

Given a terminal $v$ with a height increment $h^+\!(v)$, STAR-TC-ALGO sorts the terminal's neighbors $u_1, \dots, u_s$ by the cost of the neighbors' minimum height increments $cT\!ower^+\!(h(u_1)), \dots, cT\!ower^+\!(h(u_s))$, that cover edges $(v, u_1), \dots, (v, u_s)$. Then it returns the list of selected neighbors achieving lowest cost-to-benefit ratio associated with $v$ and $h^+\!(v)$. In each iteration of TC-ALGO, the terminal with the height increment and some of its neighbors with minimum sufficient height increments achieving the lowest cost-to-benefit ratio over all combinations are selected to add to $h$. TC-ALGO iteratively choose the height increments until the number of components is reduced to one.\footnote{The benefit, standing for the number of reduced components, is at least one. This guarantees the progress of component reduction in each iteration.} Moreover, components covered by the height function $h$ are denoted by $cover(h)$. At the beginning, $cover(h)$ stands for all isolated terminals due to zero height increment at all terminals.

In the TC problem all vertices are considered terminals and assigned tower heights. However, in SteinerTC problem some of the vertices considered non-terminals with fixed tower heights can work as relays between two terminals to further reduce the tower heights at terminals if applicable. To find the lowest cost-to-benefit ratio in SteinerTC problem, for a given terminal $v$ with a height increment, we search not only the adjacent terminals of $v$ but the terminal $u$ that can be connected through some non-terminals $v_1, \dots ,v_r$ by the path $(v, v_1, \dots ,v_r, u)$. We use the term of \textit{logical neighbor} of a given vertex $v$ to denote a terminal $u$ that can be connected by a single edge $(u, v)$ or a path $(v, v_1, \dots ,v_r, u)$ where only $v$ and $u$ are terminals.

STAR-SteinerTC-ALGO, presented in Algorithm~\ref{STAR-SteinerTC-ALGO}, is adapted from STAR-TC-ALGO. (Note that the main greedy routine, SteinerTC-ALGO, to SteinerTC problem remains the same as TC-ALGO.\footnote{In each iteration, the greedy algorithm of SteinerTC-ALGO  considers different height increments only for terminals, since the tower heights of non-terminals are fixed.}) In STAR-SteinerTC-ALGO, line 4 collects the logical neighbors of $v$ that are not in the same component as $v$ in $cover(h)$, since the benefit of $v$ cannot be improved by increasing the height of the vertices in the same component of $v$. The for loop from line 5 to 13 determines the smallest height increment of each vertex $u$ in $nbr_{logic}$ that makes $u$ connected to $v$. If there is a single edge $(u, v)$ between $u$ and $v$, we directly assign the smallest height increment at $u$ to cover this edge. Otherwise, we search a minimum cost path\footnote{The minimum cost path between $v$ and $u$ can be found in polynomial time since the heights of $v$ and all non-terminals are given.} from $v$ to $u$ including only non-terminals, $v_1, \dots, v_s$, except $u$ and $v$ and then assign the smallest height increment at $u$ to cover $(v_s, u)$. Note  that the tower heights at $v_1, \dots ,v_s$ must be high enough to cover edges $(v, v_1), \dots, (v_{s-1}, v_s)$. A vertex with the lowest cost increment (and well as the lowest height increment) is chosen from the same component and its vertices are sorted in ascending order of the cost increment from line 14 to 17. Note also that each element in $L$ reduces the number of components in $cover(h)$ by exact one. Then we can find the lowest cost-to-benefit ratio by selecting of first $k$ elements in $L$. The for loop  from line 19 to 24 determines the size of $k$ to obtain the lowest cost-to-benefit ratio $r'_{best}$. The remaining lines add the corresponding height increments for $v$ and the chosen logical neighbors of $v$ to the height increment function $incr$ and then return it. 

The performance bound of our greedy algorithm is shown in Theorem~\ref{apporxSteinerTC}, whose proof is relegated to Appendix \ref{FirstAppendix}. 
\begin{theorem}
The tower cost associated with the height function $h$ determined by SteinerTC-ALGO is at most $2 \ln{|A|}O\!P\!T$. 
\label{apporxSteinerTC} 
\end{theorem}
\begin{algorithm}
\SetAlgoLined
Input: $G = (V, E)$, terminal set $A$, non-terminal set $B$, height function $h$, vertex $v$, height increment $\delta$ at $v$\\
Output: cost-benefit-ratio $r'_{best}$, height increment function $incr$\\
$cT\!ower^+\!(h(v))$ := $cT\!ower(h(v)+\delta)$ - $cT\!ower(h(v))$\;
$nbr_{logic}$ := $\{u \mid v$, $u \in A$, $u$ is not in the same component as $v$ in $cover(h)$, there exists a path $p$ = $(v, v_1, \dots ,v_r, u)$ from $v$ to $u$ in $G$ s.t. $|p| \geq 2$, $v_1, \dots, v_r \in B$, heights $(h(v) + \delta)$ at $v$ and heights ${h(v_1), \dots, h(v_r)}$ cover edges $(v, v_1), \dots, (v_{r-1}, v_r)$\;

\For{each vertex $u \in nbr_{logic}$}{
\eIf{$(u, v) \in E$}{
$h^+\!(u)$ := smallest $\beta$ s.t. heights $(h(v) + \delta)$ at v and $(h(u) + \beta)$ at $u$ cover edge $(u, v)$\;
}{
search the path $p_{best}$ = $(v, v_1, \dots,v_s, u)$ from $v$ to $u$ with lowest cost increment\;
$h^+\!(u)$ := smallest $\beta$ s.t. heights $(h(v) + \delta)$ at v and $(h(u) + \beta)$ at $u$ cover path $p_{best}$ = $(v, v_1, \dots,v_s, u)$\;
}
$cT\!ower^+\!(h(u))$ := $cT\!ower(h(u)$ + $h^+\!(u))$ - $cT\!ower(h(u))$\;
}
$L$ := list of vertices in $nbr_{logic}$ in ascending order of $cT\!ower^+$\;
\For{each component $D \in cover(h)$}{
remove from $L$ all vertex $u \in D \cap nbr_{logic}$ except the one with lowest $cT\!ower^+$\;
}
$r'_{best}$ := $\infty$; $k_{best}$ := 0\;
\For{$1 \leq k \leq |L|$}{
$r'_{tmp}$ := $\frac{cT\!ower^+\!(h(v)) + \sum_{i=1}^{k}{cT\!ower^+\!(h(L[i]))}}{k}$\;
\If{$r'_{tmp} < r'_{best}$}{
$k_{best}$ := $k$; $r'_{best}$ := $r'_{tmp}$\;
}
}
\For{$u \in V$}{ $incr(u)$ := 0\;}
\For{$u \in L[1 \dots k_{best}]$}{$incr(u)$ := $h^+\!(u)$\;}
$incr(v)$ := $\delta$\;
\Return $(r'_{best}, incr)$\;
\caption{STAR-SteinerTC-ALGO}
\label{STAR-SteinerTC-ALGO}
\end{algorithm}

\vspace*{-5mm}
\subsection{Solution to CND Problem}
The approximation algorithm presented in \cite{hassin2004approximation} installs capacity per a given minimum cost Steiner tree to route the traffic flows from terminals to the landline. The algorithm first partitions terminals into groups according to their throughput demands. Note that the total aggregated traffic in a group is upperbounded by the cable capacity. For each group, it chooses the vertex with shortest path to the landline as a hub. The algorithm of the grouping guarantees that any internal aggregated traffic in a group does not exceed $U - dem(v_h)$, where $U$ is the uniform cable capacity and $v_h$ is the hub in this group. Thus, installing only one cable over every edge in the minimum cost Steiner tree can support the traffic flows from terminals to hubs. This capacity installation makes the connectivity in the Steiner tree.

Moreover, for each group, one cable is installed on per edge over the shortest path from its hub to the landline if additional capacity is needed. This capacity installation guarantees the bandwidth requirement from hubs to the landline. Hence, the capacity installation from terminals to hubs and from hubs to the landline satisfies the throughput constraints. Note that the cost of capacity installation from terminals to the landline is dominated by their hop distances to the landline. In our solution to CND problem, we take advantage of the grouping scheme in \cite{hassin2004approximation} on the minimum cost Steiner tree, $G_{SteinerT\!C}$, determined by SteinerTC-ALGO. For each group, the vertex with shortest hop distance to the landline $LN$ is selected as a hub. 

For the capacity installation, we first install one wireless link per edge in $G_{SteinerT\!C}$. Second, to route the aggregated traffics of groups to $LN$, one additional wireless link, if necessary, is installed per edge on the paths from hubs to $LN$, which are formed by the corresponding edges in $G_{SteinerT\!C}$. In our solution, the route function $rt$ and the traffic flow function $f$ are determined by the grouping scheme and the terminals demands.

\subsection{Performance Analysis} 
In this section, we will derive the cost performance ratio of our approximation algorithm to MNO problem by analyzing the total cost of tower and wireless link installation ($\sum_{v\in A}{cT\!ower(h(v))}+\sum_{e\in E_{SteinerT\!C}}{cLink(\lceil \frac{f(e)}{U} \rceil)}$). In the following, let $cT\!ower(SteinerT\!C)$ and $cLink(C\!N\!D)$ respectively denote the tower cost and wireless link cost in our solution.

First, we analyze the cost of tower deployment. Let $cT\!ower(O\!P\!T)$ be the cost of towers in the optimal solution to SteinerTC problem. Suppose there exists a minimum cost spanning tree where all vertices are terminals. The tower cost of this spanning tree is denoted by $cT\!ower(M\!S\!T)$. Since our solution to SteinerTC problem exploits non-terminals with existing towers to further minimize the cost of towers, $cT\!ower(SteinerT\!C) \leq cT\!ower(M\!S\!T)$. In addition, since our solution to SteinerTC problem has the cost bound ratio of $2\ln{|A|}$ to the optimal solution, we get the equation:
\begin{equation}
\begin{array}{l}
\ \ \ cT\!ower(O\!P\!T) \leq cT\!ower(SteinerT\!C) \\
\leq 2\ln{|A|}cT\!ower(O\!P\!T)
\label{towereq}
\end{array}
\end{equation}

Second, we analyze the cost of wireless link installation. Let $cLink(O\!P\!T)$ be the cost of all wireless links in the optimal solution to CND problem. Since $cLink(C\!N\!D)$ consists of the cost of installed links from terminals to hubs and the cost of installed links from hubs to the landline, we use $cLink(C\!N\!D_{T2H})$ and $cLink(C\!N\!D_{H2L})$ to respectively denote the two costs. Let $cLink(M\!S\!T_{T2H})$ be the cost of link installation from terminals to hubs in MST where each edge connects two terminals and the number of edges, $|A|-1$, is minimal. Thus, we know $cLink(M\!S\!T_{T2H}) \leq cLink(O\!P\!T)$ and $cLink(M\!S\!T_{T2H}) \leq cLink(C\!N\!D_{T2H})$. Since the worst case of the number of edges in a Steiner tree is the number of all vertices minus one $(|A|+|B|-1)$, we get the worst-case bound of $cLink(C\!N\!D_{T2H})$:
\begin{equation*}
\begin{array}{l}
\ \ \ cLink(M\!S\!T_{T2H}) \leq cLink(C\!N\!D_{T2H}) \\
\leq \frac{|A|+|B|-1}{|A|-1} cLink(M\!S\!T_{T2H})
\label{linkT2Heq1}
\end{array}
\end{equation*}
Also, we infer the bound ratio of $cLink(C\!N\!D_{T2H})$ to $cLink(O\!P\!T)$ :
\begin{equation}
\begin{array}{l}
\ \ \ cLink(C\!N\!D_{T2H}) \leq \frac{|A|+|B|-1}{|A|-1} cLink(M\!S\!T_{T2H}) \\
\leq \frac{|A|+|B|-1}{|A|-1} cLink(O\!P\!T) \leq \frac{|A|+|B|}{|A|} cLink(O\!P\!T) \\
= (1 + \frac{|B|}{|A|})cLink(O\!P\!T) 
\label{linkT2Heq2}
\end{array}
\end{equation}

\begin{figure}[hbt!]
\centerline{\includegraphics[width=0.25\textwidth]{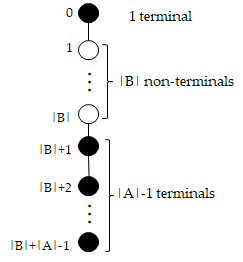}}
\caption{Worst-case topology of CND solution.}
\label{fig}
\end{figure}

Moreover, the upper bound of $cLink(C\!N\!D_{H2L})$ is obtained by considering the worst case of total hop distance from hubs to the landline. Let all terminals have the uniform demand and all wireless links have a uniform link capacity. Let $\gamma$ denote the ratio of $\frac{link\ capacity}{terminal\ demand}$. To calculate the lower bound of $cLink(O\!P\!T)$, we consider the best case of the topology, where all the vertices have 1-hop distance to the landline. In that case, a wireless link is installed on every edge from terminals to the landline. Thus, we have the best-case bound of $cLink(O\!P\!T)$, $link\ cost \times (|A|-1) \leq cLink(O\!P\!T)$. Now we calculate the worst case of the topology that is a long-chain topology where the 0th vertex is the landline (terminal), the 1st to the $(|B|)$-th vertices are non-terminals, and the $(|B|+1)$-th to $(|B|+|A|-1)$-th vertices are the remaining $|A|-1$ terminals. As Fig.~\ref{fig} depicted, the numbers labeled at left stand for the hop distances for the vertices to the landline. Note that the grouping scheme to CND problem separates all terminals vertex into $\frac{|A| \times terminal \ demand}{link\ capacity} = \frac{|A|}{\gamma}$ groups. For convenience, $\frac{|A|}{\gamma}$ is denoted as $m$. For each group, we select a hub with minimum hop distance to send the aggregated traffic equal to (or less than) the link capacity to the landline. Thus, in the 1st group containing terminals with the labels, $0, |B|+1, \dots, |B|+\gamma-1$, we select vertex 0 as a hub. In the 2nd group containing terminals with the labels, $|B|+\gamma, |B|+\gamma+1, \dots, |B|+2\gamma-1$, we select vertex $|B|+\gamma$ as a hub. In the last group containing terminals with the labels, $B+(m-1)\gamma, B+(m-1)\gamma+1, \dots, B+m\gamma-1)$, we select vertex $B+(m-1)\gamma$ as a hub. Now the total hop distance from hubs to the landline in the worst case can be presented as follows:
\begin{equation*}
\begin{array}{l}
\ \ \ 0 + (|B|+\gamma) + \dots + (|B|+(m-1)\gamma) \\
= \gamma [(\frac{|B|}{\gamma}+1) + \dots + (\frac{|B|}{\gamma}+m-1)] \\
= \gamma [(m-1)\frac{|B|}{\gamma} + (1 + \dots+ (m-1))] \\
= \gamma [(m-1)\frac{|B|}{\gamma} + \frac{m(m-1)}{2}] \\
= (m-1)(|B| + \frac{m\gamma}{2}) \\
= (\frac{|A|}{\gamma}-1)(|B| + \frac{|A|}{2}), m \leftarrow \frac{|A|}{\gamma} \\
\leq \frac{|A|}{\gamma}(|B| + \frac{|A|}{2}) = \frac{|A|^2 + 2|A||B|}{2\gamma}
\label{linkH2Leq1}
\end{array}
\end{equation*}
According to the best-case bound of $cLink(O\!P\!T)$ and the worst-case bound of $cLink(C\!N\!D_{H2L})$, we know $link\ cost \times (|A|-1) \leq cLink(O\!P\!T)$ and $link\ cost \times (|A|-1) \leq cLink(C\!N\!D_{H2L}) \leq link\ cost \times \frac{|A|^2 + 2|A||B|}{2\gamma}$.
Hence, the performance ratio of $cLink(C\!N\!D_{H2L})$ to $cLink(O\!P\!T)$ can be stated as follows: 
\begin{equation}
\begin{array}{l}
\ \ \ cLink(C\!N\!D_{H2L}) \leq \frac{|A|^2 + 2|A||B|}{2\gamma(|A|-1)}cLink(O\!P\!T) \\
\leq \frac{|A|^2 + 2|A||B|}{\gamma|A|}cLink(O\!P\!T) \leq \frac{|A| + 2|B|}{\gamma}cLink(O\!P\!T)
\label{linkH2Leq2}
\end{array}
\end{equation}
Note that $|A| - 1 \geq \frac{|A|}{2}$ due to $|A| \geq 2$.

Finally, we calculate the cost performance ratio of our approximation algorithm to MNO problem. Let $C_{O\!P\!T}$ be the total infrastructure cost of the optimal solution to MNO problem. Since the sum of the costs of the respective optimal solutions in SteinerTC and CND problems is the cost lower bound for $C_{O\!P\!T}$, we know $ cT\!ower(O\!P\!T) + cLink(O\!P\!T) \leq C_{O\!P\!T}$.
Let $C_{A\!L\!G\!O}$ be the total infrastructure cost of our solution to MNO problem, where $C_{A\!L\!G\!O} = cT\!ower(SteinerT\!C) + cLink(C\!N\!D_{T2H}) + cLink(C\!N\!D_{H2L})$. Now we are going to the discuss the three following cases of terminal demand and articulate their performance ratio.

\textit{Case 1}: The total demand from terminals is not larger than link capacity. In this case, $C_{A\!L\!G\!O} = cT\!ower(SteinerT\!C) + cLink(C\!N\!D_{T2H})$ since installation of a single wireless link on each edge from terminals to hubs is sufficient to convey the total traffic from all terminals. According to \eqref{towereq} and \eqref{linkT2Heq2}, we have 
\begin{equation*}
\begin{array}{l}
\ \ \ cT\!ower(SteinerT\!C) \\
\leq 2\ln{|A|}cT\!ower(O\!P\!T) \leq 2\ln{|A|}C_{O\!P\!T} 
\label{perfcase1eq1}
\end{array}
\end{equation*}
and
\begin{equation*}
\begin{array}{l}
\ \ \ cLink(C\!N\!D_{T2H}) \\
\leq (1 + \frac{|B|}{|A|})cLink(O\!P\!T) \leq (1 +\frac{|B|}{|A|})C_{O\!P\!T} 
\label{perfcase1eq2}
\end{array}
\end{equation*}
Hence, the cost performance ratio of case 1 is
\begin{equation}
\begin{array}{l}
\frac{C_{A\!L\!G\!O}}{C_{O\!P\!T}} \leq (1+2\ln{|A|}+ \frac{|B|}{|A|})
\label{perfcase1eq3}
\end{array}
\end{equation}

\textit{Case 2}: The total demand from terminals is larger than link capacity and the terminal demands are uniform. In this case, $C_{A\!L\!G\!O} = cT\!ower(SteinerT\!C) + cLink(C\!N\!D_{T2H}) + cLink(C\!N\!D_{H2L})$ According to \eqref{linkH2Leq2}, we know 
\begin{equation*}
\begin{array}{l}
\ \ \ cLink(C\!N\!D_{H2L}) \\
\leq \frac{|A| + 2|B|}{\gamma}cLink(O\!P\!T) \leq \frac{|A| + 2|B|}{\gamma}C_{O\!P\!T}
\label{perfcase2eq1}
\end{array}
\end{equation*}
Hence, the cost performance ratio of case 2 is
\begin{equation}
\begin{array}{l}
\ \ \ \frac{C_{A\!L\!G\!O}}{C_{O\!P\!T}} \leq (3+\ln{|A|}+ \frac{|B|}{|A|} + \frac{|A| + 2|B|}{\gamma}) \\
= O(\ln{|A|}+ \frac{|B|}{|A|} + \frac{|A| + |B|}{\gamma})
\label{perfcase2eq2}
\end{array}
\end{equation}

\textit{Case 3}: The total demand from terminals is larger than link capacity and the terminal demands are non-uniform. In this case, we can reuse the result from \eqref{perfcase2eq2} by letting $\gamma$ be $\frac{link\ capacity}{maximum\ demand}$.

\section{Cost Minimization Heuristic \\ for Hybrid Mesh Networks}
In this section, we present a heuristic for adding p2mp or omnidirectional hyperlinks, of say WiFi for p2mp or TVWS for omnidirectional, if and only if they reduce the total cost while preserving the satisfaction of throughput constraints. The algorithm iteratively finds a vertex that has some links with residual capacity and then replaces those links with a hyperlink while considering cost reduction. Below, we first articulate a concept of link replacement.  Second, we present a replacement scheme for the case of p2mp hyperlinks that substitute p2p links without changing the topology $G_{SteinerT\!C}$, tower heights, and the traffic routing.  Third, we introduce a replacement scheme for the case of omnidirectional hyperlink that substitute some p2p links and for altering the corresponding tower heights adapting to omnidirectional antennas and their subordinate directional antennas when that can offer cost reduction. Finally, we discuss the assignment of transit power in hybrid networks.

\subsection{Hyperlink Replacement Strategy}
The solution to the CND problem that we have adopted involves a partitioning wherein terminals are divided into groups whose aggregated traffics do not exceed the link capacity. This scheme results in links within groups that are not fully utilized. This underutilized capacity of wireless link is called {\em residual capacity}. For the traffic flow function $f$ returned from MNO problem, it computes the required traffic for each edge $e\! \in \! E_{SteinerT\!C}$. Let $U$ be the link capacity over its spectrum. We assume that a p2p link (which has two p2p antennas) and a p2mp hyperlink (which as one p2mp antenna and one or more p2p antennas) have the same capacity $U$, and yet an omnidirectional hyperlink (which as one omnidirectional antenna and one or more its subordinate directional antennas) can have different capacity $U_{O\!m\!n\!i}$. The residual capacity of edge $e$ is thus defined as $U \lceil \frac{f(e)}{U} \rceil – f(e)$. 

We now describe our cost minimization heuristic that replaces underutilized edges with p2mp and omnidirectional hyperlinks. Note that one edge is installed with one or more p2p links. The related notations are listed in Table~\ref{notationANT}. Recall that although the main goal of hyperlink replacement is to minimize the infrastructure cost in hybrid networks, the replacement must satisfy the throughput constraints. The heuristic does not change the MNO solution topology or the traffic routing, since changing the topology will alter the traffic flows and that may break the constraints that were 
guaranteed by solving CND. In other words, replacement with a p2mp / omnidirectional hyperlink only happens at ``candidate'' vertices with respect to some/all of their respective neighbors; it does not include any vertex which is at a distance of more than one hop from any candidate vertex. 

More specifically, given a candidate vertex $v$, the replacement only involves the edges connected to its children instead of the edge $(v, v.parent)$ that conveys the aggregated traffic of all its children. That is, we choose to not deploy a hyperlink that covers both $(v, v.parent)$ and the edges from $v$ to its children. The reasons for this are: First, for p2mp antennas, the replacement of $(v, v.parent)$ may undermine throughput performance because only one link in a p2mp hyperlink can transmit at a given time. Second, for omnidirectional links, such replacement occupies twice the capacity over the spectrum but does not effectively reduce the cost of tower height, which is the dominant factor for infrastructure cost. At $v$, we apply its height to $H\!T\!O\!m\!n\!i$, which may increase  cost ($h(v) < H\!T\!O\!m\!n\!i$) or reduce cost ($h(v) > H\!T\!O\!m\!n\!i$). Nevertheless, at vertex $v.parent$, we apply its height to $max(h(v.parent), H\!T\!O\!m\!n\!i\!S\!D)$ because $v.parent$ still has other induced edges(s). Note that $h(v)$ is likely larger than $H\!T\!O\!m\!n\!i\!S\!D$ (10m) and the available channels in the omnidirectional link's spectrum are limited and location-dependent. Therefore, allowing replacement only at children vertices offers the chance to replace more p2p links as well as reduce the corresponding tower heights.

\begin{table}[htbp]
\caption{Notations for Cost Minimization in Hybrid Networks}
\begin{center}
\begin{tabular}{c|c}
\hline
\textbf{Symbol}&\textbf{Meaning} \\
\hline
$U$ &p2p link / p2mp hyperlink capacity \\
\hline
$R_{M\!P}$&maximal transmit distance of $M\!P$ antenna, $R_{M\!P} \leq R$ \\
\hline
$B\!W$&tunable beamwidth of $M\!P$ antenna, $B\!W \leq B\!W\!M\!A\!X$ \\
\hline
$B\!W\!M\!A\!X$&maximal beamwidth of $M\!P$ antenna \\
\hline
$R\!A\!D_{M\!P}$&tunable radius of $M\!P$ antenna, $R\!A\!D_{M\!P} \leq R_{M\!P}$\\
\hline
$X$&set of vertices covered by a hyperlink \\
\hline
$cAntenna$& antenna cost function \\
&$cAntenna: \{P\!P, M\!P, O\!m\!n\!i, O\!m\!n\!i\!S\!D\} \rightarrow \mathbb{R}$ \\
\hline
$rng_{M\!P}$& MP range verification function \\
&$rng_{M\!P}: E_{SteinerT\!C} \!\times\! V_{SteinerT\!C} \!\times\! V_{SteinerT\!C}$ \\
&$\!\times \mathbb{R} \!\times\! \mathbb{R} \rightarrow \{0, \! 1\}$\\
\hline
$U_{O\!m\!n\!i}$ &omnidirectional hyperlink capacity \\
\hline
$R_{O\!m\!n\!i}$&maximal transmit distance of omnidirectional antenna \\
\hline
$R\!A\!D_{O\!m\!n\!i}$&
radius of omnidirectional antenna, $R\!A\!D_{O\!m\!n\!i} \leq R_{O\!m\!n\!i}$\\
\hline
$H\!T\!O\!m\!n\!i$&fixed tower height of omnidirectional antenna \\
\hline
$H\!T\!O\!m\!n\!i\!S\!D$&fixed tower height of subordinate directional antenna \\
&of omnidirectional hyperlink \\
\hline
$rng_{O\!m\!n\!i}$& omnidirectional range verification function \\
&$rng_{O\!m\!n\!i}: V_{SteinerT\!C} \!\times\! \mathbb{R} \!\times\! V_{SteinerT\!C} \!\times\! \mathbb{R}\rightarrow \{0, \! 1\}$\\
\hline
$V_{O\!m\!n\!i}$& set of vertices with  omnidirectional antennas\\
\hline
$V_{O\!m\!n\!i\!S\!D}(v)$& set of vertices only covered by \\
&the omnidirectional hyperlink induced from $v$\\
\hline
\end{tabular}
\label{notationANT}
\end{center}
\end{table}

\subsection{MP Deployment}
For a p2mp hyperlink, we use one p2mp antenna to share the capacity $U$ among its corresponding p2p antennas. A p2mp antenna is modeled by three variables: direction, beamwidth, and radius. We denote, for a p2mp antenna is deployed at $v$, the direction of the antenna pointed towards vertex $u$ as $\overrightarrow{\rm vu}$. We assume the p2mp antenna is flexible so as to configure its beamwidth and transmit radius so that one p2mp antenna can be set to different sector ranges; this is typical of today's p2mp antennas.

For MP deployment, there are four constraints that must be guaranteed if and only if a p2mp hyperlink is deployed at $v$ to replace p2p links associated with the edges $(v, x), x \in X$, where $X$ denotes the set of selected children of $v$.
\begin{equation*}
\begin{array}{l}
Cost: cAntenna(M\!P) < cAntenna(P\!P) |X| \\
Capacity: \sum_{x\in X}{f((v, x))} \leq U \\
Range: \prod_{(v, x)\in E_{SteinerT\!C}, x\in X}{rng_{M\!P}((v, x), v, u, B\!W,} \\ 
\ \ \ \ \ \ \ \ \ \  R\!A\!D_{M\!P}) = 1 \\
Inter\!f\!erence: \sum_{e \in E_{SteinerT\!C} \setminus \{(v, x) | x\in X\}}{rng_{M\!P}(e, v, u,} \\
\ \ \ \ \ \ \ \ \ \ \ \ \ \ \ \ \ \ B\!W, R\!A\!D_{M\!P}) = 0
\label{constraintCost}
\end{array}
\end{equation*}
The cost and capacity constraints ensure that the replacement reduces the total cost and satisfies the throughput demand. 
The range verification function $rng_{M\!P}$ examines if the edge segment $(v, x)$ is covered by the p2mp antenna at $v$ and pointed towards $u$, with the beamwidth $B\!W$ and radius $R\!A\!D_{M\!P}$. The interference constraint guarantees that the deployed p2mp antenna does not affect any edge other than $(v, x), x\! \in \! X$.

We can now state the cost minimization problem of hybrid MP deployment as follows\footnote{Note that we can maximize the cost reduction by maximizing the number of replaced p2p links. The problem can be considered a maximum coverage problem \cite{cohen2008generalized}.}:
\begin{quote}
\textit{\textbf{Input}. $G_{SteinerT\!C} = (V_{SteinerT\!C}, E_{SteinerT\!C})$, the corresponding traffic flows on edges in $E_{SteinerT\!C}$.}

\textit{\textbf{Output}. A set of vertices where p2mp antennas are deployed, along with direction, beamwidth, and radius assignments such that all 4 constraints are satisfied and the cost of wireless links is minimized.}
\end{quote}
The solution in Algorithm~\ref{MP-DEPLOY-ALGO} considers candidate vertices (and their respective children) in a bottom-up fashion. For each candidate vertex it performs local MP deployment with the greedy Algorithm~\ref{MP-ANT-REPLACE} that selects the antenna configuration (direction, beamwidth, radius) containing a largest number of uncovered vertices at each iteration.
\begin{algorithm}
\SetAlgoLined
Input: $G_{SteinerT\!C} = (V_{SteinerT\!C}, E_{SteinerT\!C})$, traffic flow function $f$\\
Output: p2mp hyperlink function $linkset_{M\!P}$\\
\For{$v \in V_{SteinerT\!C}$}{
$linkset_{M\!P}$ := $\{ \}$\;
}
$R := V_{SteinerT\!C} \setminus \{ leaves \}$\;
\While{$R \neq \{\}$}{
pick a vertex $v \in R$ such that $v.depth$ is maximum\;
$linkset_{M\!P}(v)$ := MP-ANT-REPLACE($G_{SteinerT\!C}$, $f$, $v$)\;
$R$ := $R \setminus \{v\}$\;
}
\Return $linkset_{M\!P}$\;
\caption{MP-DEPLOY-ALGO}
\label{MP-DEPLOY-ALGO}
\end{algorithm}
\begin{algorithm}
\SetAlgoLined
Input: $G_{SteinerT\!C} = (V_{SteinerT\!C}, E_{SteinerT\!C})$, traffic flow function $f$, candidate vertex $v$\\
Output: deployed p2mp antenna set of $v$, $linkset_{M\!P}(v)$ = \{($u_1$, $B\!W_1$, $R\!A\!D_{M\!P1}$), $\dots$\}\\
$S$ := \{$u | u \in v.children$\}; $end$ := $false$\;
\While{$S \neq \{\} \wedge end \neq true$}{
$index$ := 0\;
\For{$u \in S$}{
$B\!W$ := $B\!W\!M\!A\!X$; $R\!A\!D_{M\!P}$ := $R_{M\!P}$;
$X$ := \{$x | x \in v.children \wedge  rng_{M\!P}((v, x), v, u, B\!W, R\!A\!D_{M\!P}) = 1$\}\; 
\While{\big($\frac{cAntenna(M\!P)}{cAntenna(P\!P)} \! < \! |X| \; \wedge \;  (\sum_{x \in X}{f((v, x))} \! \leq \! U)$ $\wedge$ $\; \! (\!\sum_{e \in \! (\!E_{SteinerT\!C}\! \setminus \! \{\!(v, x) | x\in \! X\!\}}\!\!{rng_{\!M\!P}\!(e,\!v,\!u,\!B\!W\!\!,\!R\!A\!D_{\!M\!P}\!)\!\!=\!\!0}\!)\! $\big)
$\neq true \; \; \wedge \; \; |X| > 1$}{
decrease $B\!W$ to the closest vertex $x \in X$ (ignore any vertex at the boundary)\;
$X$ := \{$x | x \in v.children \; \wedge \;  rng_{M\!P}((v, x), v, u, B\!W, R\!A\!D_{M\!P}) = 1$\}\; 
decrease $R\!A\!D_{M\!P}$ to the vertex $x \in X$ s.t. $dist(v, x)$ is maximum\;
}
\If{$|X| > 1$}{
$L[index]$ := $(u, B\!W, R\!A\!D_{M\!P})$\;
$index$ := $index$ + 1\;
}
}
\eIf{index = 0}{
$end := true$\;
}{
$cover_{best} := 0$\;
\For{$i \gets 0, \dots, index-1$}{
$cover$ := $\#\{x | x \in v.children \; \wedge \; rng_{M\!P}((v, x), v, L[i].direction,$
$L[i].B\!W, L[i].R\!A\!D_{M\!P}) = 1\}$\;
\If{$cover > cover_{best}$}{
$cover_{best} := cover$\;
select := i\;
}
}
$linkset_{M\!P}(v) := linkset_{M\!P}(v) \cup L[select]$\;
remove the set of vertices $\{x | x \in v.children \; \wedge \; rng_{M\!P}((v, x), v, L[select].direction,$
$L[select].B\!W, L[select].R\!A\!D_{M\!P}) = 1\}$ from $S$\;
}
}
\Return $linkset_{M\!P}(v)$\;
\caption{MP-ANT-REPLACE}
\label{MP-ANT-REPLACE}
\end{algorithm}

In MP-ANT-REPLACE, the while loop from line 4 to 32 finds the p2mp antenna with a configuration (direction, beamwidth, radius) covering a maximum number of uncovered vertices at each iteration. The for loop from line 6 to 17 tries to assign the antenna direction to every vertex $u$ in the available children and then gradually decreases its beamwidth and radius until the constraints (cost, capacity, and interference) are satisfied. The tuning of antenna configuration is presented in the while loop from line 8 to 12. Let $L$ be the candidate list of the available antenna configurations to be deployed in each iteration, as assigned from line 13 to 16. After $L$ is computed, from line 18 to 31 we choose the one that covers the maximum vertices. If there is no available candidate found, from line line 18 to 19, we end the replacement at $v$. Otherwise, in the for loop from line 22 to 28 we find the candidate with the best coverage. Then we adopt this antenna configuration and remove the vertices covered by it. Note that this iteration (loop from line 4 to 32) ends at line 8, when there is no configuration satisfying the constraints. Finally, the set of selected antenna configurations is returned as $ linkset_{M\!P}(v)$.

\subsection{Omnindirectional Antenna Deployment}
For an omnidirectional hyperlink, we use one omnidirectional antenna to share its capacity, $U_{O\!m\!n\!i}$, with its corresponding subordinate directional antennas. 
Unlike p2mp antennas, an omnidirectional antenna covers its subordinate directional antennas according to its omnidirectional coverage shaped by $R_{O\!m\!n\!i}$\footnote{$R_{O\!m\!n\!i}$ varies with the coding schemes (i.e., 64QAM, 16QAM, QPSK) \cite{stevenson2009ieee}. To support broadband services, we select the one that achieves the best capacity of $U_{O\!m\!n\!i}$.}. We assume that the spectrum of the omnidirectional hyperlinks does not overlap with that of the p2p/p2mp links, and also that these links are capable of non-line-of-sight propagation; these assumptions are again common; i,e, if we use TVWS omnidirectional hyperlinks with WiFi p2p/p2mp links. The implication of the latter assumption in particular is that omnidirectional antenna and their subordinate directional antennas are mounted at a fixed height ($H\!T\!O\!m\!n\!i$ and $H\!T\!O\!m\!n\!i\!S\!D$) for non-line-of-sight propagation so the needed tower heights to omnidirectional hyperlinks are considerably shorter.

For Omnidirectional deployment, four constraints must be guaranteed if and only if an omnidirectional hyperlink at a candidate vertex $v$ is going to replace WiFi p2p links associated with the edges $(v, x), x \in X$:
\begin{equation*}
\begin{array}{l}
Cost: cAntenna(O\!m\!n\!i) + cAntenna(O\!m\!n\!i\!S\!D) |X| \\
\ \ \ \ \ \ \ \ + cT\!ower(\max(H\!T\!O\!m\!n\!i, h(v))) \\
\ \ \ \ \ \ \ \ + \sum_{u \in V_{O\!m\!n\!i\!S\!D}}{cT\!ower(H\!T\!O\!m\!n\!i\!S\!D)} \\
\ \ \ \ \ \ \ \ < 2 \times cAntenna(P\!P) |X| + cT\!ower(h(v)) \\
\ \ \ \ \ \ \ \ + \sum_{u \in V_{O\!m\!n\!i\!S\!D}}{cT\!ower(h(u))}\\
Capacity: \sum_{x\in X}{f((v, x))} \leq U_{O\!m\!n\!i} \\
Range: dist(v, x) \leq R_{O\!m\!n\!i}, (v, x) \in E_{SteinerT\!C}, x \in X \\
Inter\!f\!erence: \sum_{u \in V_{O\!m\!n\!i}}{rng_{O\!m\!n\!i}(u, u.R\!A\!D_{O\!m\!n\!i}, v,} \\
\ \ \ \ \ \ \ \ \ \ \ \ \ \ \ \ \ \ R\!A\!D_{O\!m\!n\!i}) = 0
\label{constraintCost}
\end{array}
\end{equation*}
Since Omnidirectional deployment as formulated not only affects the antenna cost but also changes the tower heights associated with the hyperlink, the cost constraint guarantees that the total infrastructure cost consisting of link installation and tower deployment must be reduced. The capacity and range constraints ensure that the terminal demands are still satisfied and the corresponding vertices are covered by the omnidirectional antenna. The interference constraint guarantees that the deployed omnidirectional antenna at the candidate vertex $v$ does not affect any other omnidirectional hyperlink. The range verification function $rng_{O\!m\!n\!i}$ examines if the candidate vertex $v$ with radius $R\!A\!D_{O\!m\!n\!i}$ overlaps any other omnidirectional antenna at $u$ with its radius $u.R\!A\!D_{O\!m\!n\!i}$, $u\! \in \! V_{O\!m\!n\!i}$.

We can now state the cost minimization problem of hybrid Omnidirectional deployment as follows:
\begin{quote}
\textit{\textbf{Input}. $G_{SteinerT\!C} = (V_{SteinerT\!C}, E_{SteinerT\!C})$, the corresponding traffic flows on edges in $E_{SteinerT\!C}$, the tower height at vertices in $V_{SteinerT\!C}$.}

\textit{\textbf{Output}. A set of omnidirectional antennas with the covered vertices deployed subordinate directional antennas such that the four constraints are satisfied and the cost of towers and wireless links is minimized.}
\end{quote}
The solution in Algorithm~\ref{OMNI-DEPLOY-ALGO} also considers each candidate vertex in a bottom up manner and tries to deploy an omnidirectional antenna for its children using Algorithm~\ref{OMNI-ANT-REPLACE}. 
The algorithm maximizes the cost reduction by maximizing the covered vertices whose tower heights are changed to $H\!T\!O\!m\!n\!i\!S\!D$. Recall that it is profitable to change a tower height to $H\!T\!O\!m\!n\!i\!S\!D$ due to low cost (i.e., only \$100 for 10m). 

\begin{algorithm}
\SetAlgoLined
Input: $G_{SteinerT\!C} = (V_{SteinerT\!C}, E_{SteinerT\!C})$, traffic flow function $f$, height function $h$\\
Output: omnidirectional hyperlink function $linkset_{O\!m\!n\!i}$\\
$V_{O\!m\!n\!i} := \{\}$\;
\For{$v \in V_{SteinerT\!C}$}{
$linkset_{O\!m\!n\!i}$ := $\{ \}$\;
}
$R := V_{SteinerT\!C} \setminus \{ leaves \}$\;
\While{$R \neq \{\}$}{
pick a vertex $v \in R$ such that $v.depth$ is maximum\;
$linkset_{O\!m\!n\!i}(v)$ := OMNI-ANT-REPLACE($G_{SteinerT\!C}$, $f$, $h$, $V_{O\!m\!n\!i}$, $v$)\;
\If{ $linkset_{O\!m\!n\!i}(v) \neq \{ \}$}{
$V_{O\!m\!n\!i} := V_{O\!m\!n\!i} \cup {v}$\;
}
$R$ := $R \setminus \{v\}$\;
}
\Return $linkset_{O\!m\!n\!i}$ \;
\caption{OMNI-DEPLOY-ALGO}
\label{OMNI-DEPLOY-ALGO}
\end{algorithm}

\begin{algorithm}
\SetAlgoLined
Input: $G_{SteinerT\!C} = (V_{SteinerT\!C}, E_{SteinerT\!C})$, traffic flow function $f$, height function $h$, omnidirectional antenna set $V_{O\!m\!n\!i}$, candidate vertex $v$\\
Output: deployed subordinate directional antennas of $v$, $linkset_{O\!m\!n\!i}(v) = \{u_1, \dots, u_{|X|}\}$\\
$X$ := \{$x | x \in v.children \wedge dist(v, x) \leq R_{O\!m\!n\!i}$\}\; 
$V_{O\!m\!n\!i\!S\!D}$ := \{$x | x \in X$, $(v, x)$ is the only induced edge from $x$\}\;
\While{\big($cAntenna(O\!m\!n\!i) + cAntenna(O\!m\!n\!i\!S\!D) |X| + cT\!ower(\max(H\!T\!O\!m\!n\!i, h(v))) + \sum_{u \in V_{O\!m\!n\!i\!S\!D}}{cT\!ower(H\!T\!O\!m\!n\!i\!S\!D)} < 2 \times cAntenna(P\!P) |X| +cT\!ower(h(v)) + \sum_{u \in V_{O\!m\!n\!i\!S\!D}}{cT\!ower(h(u))} \; \wedge \; \sum_{x \in X}{f((v, x))} \leq U_{O\!m\!n\!i}$\big) $\neq true \; \wedge \; |X| > 0$}{

\eIf{$\sum_{u \in V_{O\!m\!n\!i}}{rng_{O\!m\!n\!i}(u, u.R\!A\!D_{O\!m\!n\!i}, v, R\!A\!D_{O\!m\!n\!i}) = 0}$}{
\eIf{$|X \setminus V_{O\!m\!n\!i\!S\!D}| > 0$}{
remove $x \in X \setminus V_{O\!m\!n\!i\!S\!D}$ s.t. $f(v, x)$ is maximum, from $X$\;
}{
remove $x \in X$ s.t. $f(v, x)$ is maximum, from $X$\;
}
}{
remove $x \in X$ s.t. $dist(v, x)$ is maximum, from $X$\; 
decrease $R\!A\!D_{O\!m\!n\!i}$ to the vertex $x \in X$ s.t. $dist(v, x)$ is maximum\;
}
$V_{O\!m\!n\!i\!S\!D}$ := \{$x | x \in X$, $(v, x)$ is the only induced edge from $x$\}\;
}
\eIf{$|X| > 0$}{
 $linkset_{O\!m\!n\!i}(v)$ = X\;
}{
 $linkset_{O\!m\!n\!i}(v)$ = \{\}\;
}

\Return $linkset_{O\!m\!n\!i}(v)$\;
\caption{OMNI-ANT-REPLACE}
\label{OMNI-ANT-REPLACE}
\end{algorithm}

In OMNI-ANT-REPLACE, the while loop from line 5 to 17 finds a maximum set of vertices where we are going to substitute with subordinate directional antennas. A way to trim the set to satisfy the cost and capacity constraints is shown from line 6 to 15. First of all, line 6 verifies if the configuration of the candidate vertex $v$ and radius $R\!A\!D_{O\!m\!n\!i}$ has no overlap with any other deployed omnidirectional antenna. If there is no overlap, then line 7 examines whether any vertex not in $V_{O\!m\!n\!i\!S\!D}$ exists. We then remove the one with maximum flow at line 8. Otherwise, at line 10 we remove the one in $V_{O\!m\!n\!i\!S\!D}$ with maximum flow. This enhances the chance to satisfy capacity constraint as well as remain more residual capacity. If there exists overlap between the candidate vertex $v$ and any other omnidirectional antenna, from line 13 to 14, we remove the vertex in $X$ with maximum distance to $v$ and then retune the radius $R\!A\!D_{O\!m\!n\!i}$. After that, we determine from line 18 to 19, the covered vertices if any set $X$ found to satisfy the constraints. The set of covered vertices is then returned as $linkset_{O\!m\!n\!i}(v)$.

\subsection{Power Assignment and Interference Avoidance}
In a homogeneous network (i.e., with only p2p links), the antenna transmit power varies according to the link distance. To mitigate the interference between any two p2p links, power selection is possible using, say, the 
2P MAC protocol \cite{raman2005design} that solves interference issues more efficiently than CSMA/CA for long-distance links. In hybrid network, in addition, our algorithm for MP / Omnidirectional deployment considers the interference constraint by tuning the transmit radius (and beamwidth of p2mp antennas) to the farthest covered vertex. Hence, the transit power of antennas is assigned during the replacement and cause no additional interference.

\section{Conclusions and Future Work}
Terrestrial wireless meshes of long-distance links offer an important alternative to address the digital divide for rural areas.  Adoption of new technologies for high throughput, low latency modalities, such as optical wireless, can further improve their viability. We have focused on the middle-mile network optimization aspect of their sustainable development. We have proposed the first MNO solution that works in polynomial time.  Our solution has an approximation ratio of $O(\ln{|A|}+\frac{|B|}{|A|}+\frac{|A|+|B|}{\gamma})$ in tower and antenna cost. We also introduced the SteinerTC problem to ease the MNO problem by determining a network topology covered by minimal tower heights that assure line-of-sight links. Our MNO solution applies to hybrid networks, allowing lower-cost hyperlinks to share capacity and eliminate underutilized capacity.


Our future work includes evaluation of our solution in targeted rural areas of southern Ohio in partnership with local wireless network providers. This will include cost-analysis for likely scenarios of network growth that build on capacity flexibility offered by our solution.


\bibliographystyle{IEEEtran}
\bibliography{myref}

\appendices

\section{Approximation Ratio of SteinerTC Solution}
\label{FirstAppendix}

\begin{lemma} 
Given a height increment $\delta$ at $v$, STAR-SteinerTC-ALGO returns height increments of $v$'s logical neighbors such that the cost-to-benefit ratio of $v$ is minimized. 
\end{lemma}

\begin{proof}
Let $incr_{O\!P\!T}$ be the increment function with height increment $\delta$ at $v$ and the height increments of some of the logical neighbors of $v$. Let $J$ be the set of chosen logical neighbors of $v$ in $incr_{O\!P\!T}$ such that the cost-to-benefit ratio at $v$ is optimal. Note that the benefit of $incr_{O\!P\!T}$ at $v$ is $|J|$ that minimized the cost-to-benefit ratio. If there exists more than one vertex in $J$ in the same component $D_i$, the benefit of $incr_{O\!P\!T}$ at $v$ is less than $J$. Then we can remove one vertex in $J \cap D_i$ to obtain a lower cost-to-benefit ratio at $v$. This contradicts the definition of $incr_{O\!P\!T}$. 

Let $incr_{A\!L\!G\!O}$ be the increment function returned by STAR-SteinerTC-ALGO and $K$ be the set of chosen logical neighbors of $v$ in $incr_{A\!L\!G\!O}$. The benefit of $incr_{O\!P\!T}$ at $v$ is $|K|$ since line 16 in Algorithm~\ref{STAR-SteinerTC-ALGO} remains only one vertex in each component. Now we will compare the selection of $J$ and $K$. First, suppose $|K|$ is equal to $|J|$. Note that, for each component having at least one logical neighbor of $v$, STAR-SteinerTC-ALGO only remains one vertex with lowest cost increment. $incr_{A\!L\!G\!O}$ also considers the same set of vertices, denoted as $L$, because choosing another vertex in the component makes higher cost increment but identical benefit. Since the for loop from line 19 to 24 chooses first $|K|$ elements in $L$ with ascending order of cost increment, the sum of cost increments with size $|K|$ is minimized. Thus, $incr_{A\!L\!G\!O}$, equal to $incr_{O\!P\!T}$, is optimal.

Second, suppose $|K|$ is less than $|J|$. Since $K$ is the first $|K|$ elements in $L$ and $K$ is the set with minimal cost increment in size $|K|$, $J \cap K = K$ and $J \setminus K$ contains at least one vertex $u$ not in the components of vertices in $K$. Note that vertex $u$ is the one with minimal cost increment in its component, thus $u$ is in $L$. However, the for loop from line 19 to 24 in Algorithm~\ref{STAR-SteinerTC-ALGO} considers all sizes from 1 to $|L|$ to determine the lowest cost-to-benefit ratio. $u$ is selected in $incr_{A\!L\!G\!O}$ if and only if $u$ and the vertices with lower cost increments than $u$ in $L$ make lowest cost-to-benefit ratio. This conflicts with the assumption. Hence, $|K|$ must be identical to $|J|$. According to the first statement, $incr_{A\!L\!G\!O}$ is optimal.

Finally, suppose $|K|$ is larger than $|J|$. Since $J \cap K = J$, we can obtain a lower cost-to-benefit ratio by removing $L[|J|+1], \dots, L[|K|]$ from $K$. However, the set $L[1 \dots |J|]$ is considered by STAR-SteinerTC-ALGO and chosen as $incr_{A\!L\!G\!O}$ if and only if $L[1 \dots |J|]$ makes lowest cost-to-benefit ratio. This conflicted with the assumption. Hence, $|K|$ must be identical to $|J|$ and thus $incr_{A\!L\!G\!O}$ is optimal.
\end{proof}

\begin{lemma} 
The cost-to-benefit ratio of the height increment selected in each iteration of SteinerTC-ALGO is at most twice the cost-to-benefit ratio of any height increment centered at $v$.
\label{SteinerTC2approximation}
\end{lemma}

\begin{proof}
Since SteinerTC-ALGO remains the same as TC-ALGO and STAR-SteinerTC-ALGO returns local optimum, it also achieves a 2-approximation of the optimal height increment centered at $v$ as TC-ALGO does, which is proved in \cite{panigrahi2008minimum}.
\end{proof}

\begin{definition}
A spider is a tree containing at least two leaves and having disjoint paths from the root to all leaves. There is at most one vertex with degree larger than 2 in a spider. A root of a spider is a vertex with edge-disjoint paths to all leaves of the spider. The root is unique if and only if a spider contains more than two leaves.
\end{definition}

\begin{definition}
Let $G = (V, E)$ be an undirected connected graph and let $S$ be a subset of $V$. A spider decomposition $S\!D(G, S)$ is a set of vertex-disjoint spiders in G such that the union of leaves of spiders in $S\!D(G, S)$ contains all vertices in S.
\end{definition}

\begin{lemma}  \cite{klein1995nearly}.~~~Let $G = (V, E)$ be an undirected connected graph and let $S$ be a subset of $V$, where $|S| \geq 2$. There exist a spider decomposition $S\!D(G, S)$ in G. 
\label{SpiderDecomp} 
\end{lemma}

\begin{lemma}
Consider iteration $i$ in the greedy algorithm SteinerTC-ALGO. Let the height function at the beginning of iteration $i$ be $h_{i-1}$. Let $O\!P\!T$ be the cost of optimal height function $h_{O\!P\!T}$. The ratio of $\frac{cost}{benefit-1}$ associated with the height increment selected by iteration $i$ is at most $\frac{2O\!P\!T}{|cover(h_{i-1})|}$.
\label{SteinerTCc2b}
\end{lemma}

\begin{proof}
Let $c_i$ be the cost increment and $b_i$ be the benefit associated with the height increment determined in iteration $i$. Note that edges in $cover(h_{O\!P\!T})$ associated with the corresponding height increments can be used to connect $|cover(h_{i-1})|$ components into a tree. Let $cover(h_{i-1})$ have $\phi_{i-1}$ components $D_1, \dots , D_{\phi_{i-1}}$. Let $T_i$ be the graph obtained from $cover(h_{O\!P\!T})$ by contracting $D_j \cap cover(h_{O\!P\!T}), j = 1, \dots, \phi_{i-1} $ to a supervertex. That is to say, starting with $cover(h_{O\!P\!T})$, for each component $D_j$, $D_j \cap cover(h_{O\!P\!T})$ is merged into a supervertex. All the supervertices will be connected in one single component after the final iteration is done. Note that the vertices in $D_j \setminus cover(h_{O\!P\!T})$, that are non-terminals, have already connected by the supervertex contracted from $D_j \cap cover(h_{O\!P\!T})$. Connecting all supervertices yields a Steiner tree spanning all terminals. 

Let $S$ denote the set of the $\phi_{i-1}$ supervertices. According to Lemma~\ref{SpiderDecomp}, $T_i$ contains a spider decomposition $S\!D(T_i, S)$. Each spider in $S\!D(T_i, S)$ is associated with height increments (derived from $h_{O\!P\!T}$) at its leaves (i.e., supervertices). The cost increment from $cover(h_{i-1})$ to $T_i$ is at most $O\!P\!T$. Let $r_1, \dots, r_w$ be the roots of the spiders in $S\!D(T_i, S)$ and $m_1, \dots, m_w$ be the number of supervertices in $S$ in each of these spiders. Let $sc_j$ denote the cost increment to form the spider with root $r_j$, where $m_j \geq 2$. Note that a spider with root $r_j$ spans a subset of components $D_1, \dots , D_{\phi_{i-1}}$, which is the $m_j$ components corresponding to $m_j$ supervertices in this spider. Accordingly, the cost-to-benefit ratio at $r_j$ is $\frac{sc_j}{m_j - 1}$. Note also that the cost-to-benefit ratio at $r_j$, a non-terminal with fixed height, can be considered the cost-to-benefit ratio at a terminal $v$ in the supervertices of the same spider, that is induced by an edge $(v, u)$, $u\! \in \!B$, in the spider associated with the height increment $h^+(v)$. Therefore, the cost-to-benefit ratio at $v$ in the spider is at least the cost-to-benefit ratio determined by STAR-SteinerTC-ALGO in a consideration of the same terminal $v$ and height increment $h^+(v)$.

Since SteinerTC-ALGO chooses a vertex with lowest cost-to-benefit ratio in each iteration, for each spider in $S\!D(T_i, S)$, $\frac{c_i}{b_i} \leq \frac{sc_j}{m_j - 1},  j = 1, \dots, w$. Moreover, $\frac{c_i}{b_i}$ is at most the average cost-to-benefit ratio over spiders in $S\!D(T_i, S)$, thus we have:
\begin{equation*}
\begin{array}{l}
\ \ \ \ \ \ \ \frac{c_i}{b_i} \leq \min_{j}{\frac{sc_j}{m_j-1}} \leq \frac{\sum_{j=1}^{w}{sc_j}}{\sum_{j=1}^{w}{(m_j - 1)}} \\
\implies \frac{c_i}{b_i+1} \leq \min_{j}{\frac{sc_j}{m_j}} \leq \frac{\sum_{j=1}^{w}{sc_j}}{\sum_{j=1}^{w}{m_j}} \\
\implies \frac{c_i}{b_i+1} \sum_{j=1}^{w}{m_j} \leq \sum_{j=1}^{w}{sc_j} \leq O\!P\!T
\label{spidereq1}
\end{array}
\end{equation*}
Due to $\sum_{j=1}^{w}{m_j} = \phi_{i-1}$, 
\begin{equation}
\begin{array}{l}
\frac{c_i}{b_i+1} \leq \frac{O\!P\!T}{\phi_{i-1}} = \frac{O\!P\!T}{|cover(h_{i-1})|}
\label{spidereq2}
\end{array}
\end{equation}
Note that $c_i \leq sc_j$. According to Lemma~\ref{SteinerTC2approximation}, a single iteration of the greedy algorithm SteinerTC-ALGO finds a 2-approximation of the optimal height increment, the minimum ratio of $\frac{cost}{benefit-1}$ for the height increment selected by SteinerTC-ALGO is thus at most $2 \times \frac{O\!P\!T}{|cover(h_{i-1})|} = \frac{2O\!P\!T}{|cover(h_{i-1})|}$.
\end{proof}

From Lemma~\ref{SteinerTCc2b}, the approximation ratio of $2\ln{|A|}O\!P\!T$ in Theorem~\ref{apporxSteinerTC} is shown as follows.
According to the definition of benefit, $\phi_i = \phi_{i-1} - b_i$ in iteration $i$. Substituting  Eq.~\ref{spidereq2} into this equation, we obtain
\begin{equation*}
\begin{array}{l}
\phi_i = \phi_{i-1} - b_i \leq \phi_{i-1} - \frac{b_i+1}{2} \leq  \phi_{i-1}(1 - \frac{c_i}{2OPT})
\end{array}
\end{equation*}
Let the total number of iterations taken in SteinerTC-ALGO be $p$ and $\phi_{p}=1$. We have 
\begin{equation*}
\begin{array}{l}
\ \ \ \ \ \ \ \phi_{p}=\phi_{0}\prod_{j=1}^{p}{(1 - \frac{c_j}{2OPT})} \\
\implies  0 = \ln{\frac{\phi_{0}}{\phi_{p}}} + \ln{\prod_{j=1}^{p}{(1 - \frac{c_j}{2OPT})}}
\end{array}
\end{equation*}
By using $\ln{(1+x)}\leq \ln{x}$, we get 
\begin{equation*}
\begin{array}{l}
\ \ \ \ \ \ \  0 = \ln{\frac{\phi_{0}}{\phi_{p}}} + \ln{\prod_{j=1}^{p}{(1 - \frac{c_j}{2OPT})}} \\
\ \ \ \ \ \ \ \ \ \leq \ln{\frac{\phi_{0}}{\phi_{p}}}+ \sum_{j=1}^{p}{\frac{-c_j}{2OPT}} \\
\implies \frac{1}{2OPT}\sum_{j=1}^{p}{c_j} \leq \ln{\frac{\phi_{0}}{\phi_{p}}} \\
\implies \sum_{j=1}^{p}{c_j} \leq 2\ln{|A|}OPT
\end{array}
\end{equation*}
This proves that the total cost for the height function $h$ returned by SteinerTC-ALGO is at most $2\ln{|A|}OPT$.


\end{document}